\documentclass[lettersize,journal]{IEEEtran}
\usepackage{amsmath,amsfonts}
\usepackage{algorithmic}
\usepackage{algorithm}
\usepackage{array}
\usepackage{textcomp}
\usepackage{stfloats}
\usepackage{url}
\usepackage{verbatim}
\usepackage{graphicx}
\usepackage{float}
\usepackage{cite}
\usepackage{amsthm}
\usepackage{subfigure}
\usepackage{caption}
\newtheorem{lemma}{Lemma}
\newtheorem{theorem}{Theorem}
\newtheorem{corollary}{Corollary}
\usepackage{amssymb}
\usepackage{color}
\usepackage{indentfirst}
\begin{document}

	\title{Coverage Analysis for Cellular-Connected Random 3D Mobile UAVs with Directional Antennas}

	\author{
		Hongguang~Sun, \emph{Member, IEEE,} 
		Chao~Ma, 
		Linyi Zhang, 
		Jiahui Li, 
		Xijun~Wang, \emph{Member, IEEE,} \\ 
		Shuqin~Li, 
		and Tony Q.S. Quek, \emph{Fellow, IEEE}
		\vspace{-1em}
	}

	\maketitle

	\begin{abstract}
		This letter proposes an analytical framework to evaluate the coverage performance of a cellular-connected unmanned aerial vehicle (UAV) network in which UAV user equipments (UAV-UEs) are equipped with directional antennas and move according to a three-dimensional (3D) mobility model. The ground base stations (GBSs) equipped with practical down-tilted antennas are distributed according to a Poisson point process (PPP). With tools from stochastic geometry, we derive the handover probability and coverage probability of a random UAV-UE under the strongest average received signal strength (RSS) association strategy. The proposed analytical framework allows to investigate the effect of UAV-UE antenna beamwidth, mobility speed, cell association, and vertical motions on both the handover probability and coverage probability. We conclude that the optimal UAV-UE antenna beamwidth decreases with the GBS density, and the omnidirectional antenna model is preferred in the sparse network scenario. What's more, the superiority of the strongest average RSS association over the nearest association diminishes with the increment of GBS density.
	\end{abstract}


	\vspace{-0.5em}
	\section{Introduction}

	In recent years, cellular-connected unmanned aerial vehicle (UAV) has attracted more and more attention from both industry and academia, where the ubiquitous cellular networks have served as the communication technology to provide reliable connectivity [1][2]. However, equipped with down-tilted antennas, the ground base stations (GBSs) have been optimized for ground users, while serving UAV user equipments (UAV-UEs) with only weak side-lobes of antennas. What's worse, the three-dimensional (3D) mobility of UAV-UEs, especially the variation in flight altitude, results in frequent handover \cite{ref3}.

	To boost the development of cellular-connected UAV technology, it is of great importance to evaluate the impact of 3D mobility on network performance. Although several analytical frameworks have been proposed to assess the performance of cellular-connected UAV networks, most of them considered the static UAV-UE scenario and focused on the impact of down-tilted antennas, LoS/NLoS components, or UAV-UE flight altitude \cite{ref4}. For the mobile UAV scenarios, the authors in \cite{ref5} restricted the movement of UAV-GBSs within a finite 3D cylindrical region by using a mixed random mobility model. The authors in \cite{ref6} proved that to provide a uniform ground user coverage, the UAV-GBSs should follow certain trajectory processes. The authors in \cite{ref7} proposed several canonical mobility models to characterize the mobile UAV-GBSs in an infinite drone cellular network. However, the UAV-GBSs mobility models proposed in [5][6] are not applicable to mobile UAV-UEs, since UAV-UEs can travel very long distances and cross multiple GBSs serving areas which is far beyond a finite 3D cylindrical region. In addition, the two-dimensional (2D) mobility models proposed in \cite{ref7} for UAV-GBSs cannot capture the 3D motions of UAV-UEs. 

	\begin{figure}[!t]
		\centering
		\includegraphics[width=1\linewidth]{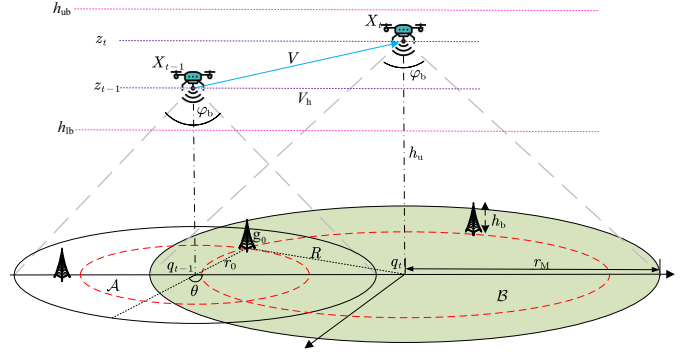}
		\captionsetup{font={scriptsize}}
		\caption{
			Illustration of a handover scenario for a 3D mobile UAV-UE equipped with directional antenna, where $h_{\mathrm{u}}$ and $\varphi_{\mathrm{b}}$ represent the UAV height and antenna beamwidth, respectively. The UAV-UE is initially located at $X_{t-1}$ at the starting waypoint of the $t$-th movement with $q_{t-1}$ being the horizontal projection, and $r_{0}$ denoting the horizontal distance between the serving GBS $g_{0}$ and the UAV-UE located at $X_{t-1}$. The UAV-UE moves to $X_{t}$ at speed $V$ (with horizontal speed $V_{\mathrm{h}}$).
		}
		\label{fig_1}
		\vspace{-1em}
	\end{figure}

	To improve the reliable connectivity for 3D mobile UAV-UEs, directional antenna technology can be served as an effective approach. By limiting the antenna beamwidth of UAV-UEs, the aggregated interference from interfering GBSs can be restricted while the higher antenna gain can be achieved, leading to the enhanced coverage probability. Although the effect of directional antenna model was quantified in \cite{ref4}, the 3D motions of UAV-UEs were not captured. The authors in \cite{ref8} proposed a 3D random way-point (RWP) mobility model, and obtained the rigorous coverage analysis of coordinated multipoint (CoMP) transmissions for 3D mobile UAV-UEs equipped with omnidirectional antennas under the nearest association strategy. The analytical complexity gets even worse when the strongest average received signal strength (RSS) association strategy is adopted. Because of the LoS/NLoS air-to-ground (A2G) transmissions, the serving GBS may not be the nearest one. When a handover event occurs caused by the 3D mobility of a UAV-UE, there may exist four possibilities: the UAV-UE may handover from a LoS/NLoS serving GBS to a LoS/NLoS target GBS.

	Motivated by the aforementioned, in this letter, we aim to quantify the superiority of the directional antenna technology in improving the coverage performance of a random 3D mobile UAV-UE by adopting the strongest average RSS association strategy. With tools from stochastic geometry, we derive the exact analytical expressions of handover probability and coverage probability. Moreover, the impacts of some key system parameters on the network performance are evaluated.

	\section{SYSTEM MODEL}

	\subsection{Network Model}
		We consider a downlink cellular-connected UAV network composed of GBSs and UAV-UEs as shown in Fig. \ref{fig_1}. The GBSs are assumed to have a fixed transmit power $\mathrm{P_{t}}$ and a constant height $h_{\mathrm{b}}$, whose two-dimensional (2D) spatial locations are distributed according to a homogeneous Poisson point process (HPPP) $\Phi$ with density $\mit\lambda$$_{\mathrm{b}}$. We consider the practical GBS antenna structure, which is implemented by multiple sector antennas, leading to the vertically directional and horizontally omnidirectional radiation pattern. Since GBS antennas are typically down-tilted to provide sufficient coverage for ground users \cite{ref9}, UAV-UEs can be only served by GBSs through sidelobe with gain $\mathrm{G_{b}}$. UAV-UEs are initially distributed according to another point process and move according to a random three-dimensional (3D) mobility model which is detailed subsequently. Without loss of generality, we focus on a typical UAV-UE for performance analysis.
		The altitude of the typical UAV-UE is denoted by $h_{\mathrm{u}}$ which is dependent on the vertical motion pattern. The UAV-UE is equipped with a directional antenna of which the beamwidth $\varphi_{\mathrm{b}}$ can be adjusted via a mechanical or electrical mechanism with the main-lobe gain being approximately by $\mathrm{G_{u}}=29000/\varphi_{\mathrm{b}}^2$, and the gain outside of main-lobe being 0\cite{ref8}. Therefore, the total impact of GBS and UAV-UE antenna gains is defined as $\mathrm{G_{tot}=G_{b}G_{u}}$.

	\subsection{Mobility Model}

		We adopt the 3D RWP model proposed in \cite{ref8} to represent the 3D movement process of UAV-UEs. The model defines the vertical movement of UAV-UEs in the limited height area, which is used to model the scenario where UAV-UEs must alternate their altitudes by making vertical motions as shown in Fig. \ref{fig_1}.
		We define triples $\left\{  X_{t-1},X_{t},V \right\},\forall X_{t} \in \mathbb{R}^{3},t \in \mathbb{N}$ to represent the movement track of the typical UAV-UE, where $t$ is the sampling interval of the track point, $X_{t-1}$ and $X_{t}$, respectively, represents the starting and end way-points of the $t$-th movement and $V$ is the UAV-UE speed which is assumed to be constant. We assume that the horizontal speed $V_{\mathrm{h}}$ of the UAV-UE from $X_{t-1}$ to $X_{t}$ is a random variable (RV), where $V_{\mathrm{h}} = \frac{V\varrho_{t} }{\sqrt{\varrho_{t} ^{2}+\left ( z_{t} - z_{t-1}  \right ) ^{2}  } } $. Particularly, similar to [8], the horizontal transition lengths ${\rho_{1}, \rho_{2},... }$ are chosen to be independent and identically distributed (i.i.d.), which follows the Rayleigh distribution with mobility parameter $\mu$ resulting in the displacement probability density function (PDF) $f_{\rho_{t}}\left( \varrho_{t} \right)  = 2 \pi \mu \varrho_{t} e^{-\pi \mu \varrho_{t}^{2}}$.
		We assume that the UAV-UE height $Z_{t}$ is uniformly distributed on $\left[ {h_{\mathrm{lb}},h_{\mathrm{ub}}} \right]$, i.e., $Z_{t} \sim \mathcal{U} \left ( h_{\mathrm{lb}},h_{\mathrm{ub}} \right ) $ with $f_{Z_{t}}\left( z_{t} \right)  = \frac{1}{h_{\mathrm{ub}}-h_{\mathrm{lb}}}  , \forall h_{\mathrm{lb}} \le z_{t} \le h_{\mathrm{ub}}$.

		We assume that the typical UAV-UE can move in any direction at the horizontal projection on the $x-y$ plane with equal probability. As shown in Fig. \ref{fig_1}, $\theta$ is taken with respect to the horizontal direction of connection between the UAV-UE located at $X_{t-1}$ and its serving GBS $g_{0}$. Denoting the corresponding RV as $\Theta$, the PDF of $\Theta \in \left [0, \pi \right ]$ follows the uniform distribution given by $f_{\Theta}(\theta)=\frac{1}{\pi} $. A handover occurs when there exists another GBS providing stronger signal than its original serving GBS located at $g_{0}$.

	\subsection{Channel Model}

		We adopt a realistic air-to-ground (A2G) channel model, consisting of both large-scale path loss and small-scale fading. We consider the possibility of LoS and NLoS transmissions, and define the corresponding path loss exponents as $\alpha_{\mathrm{L}}$ and $\alpha_{\mathrm{N}}$, $2 < \alpha_{\mathrm{L}} < \alpha_{\mathrm{N}}$. According to the movement of UAV, the path loss is expressed as
		\begin{equation}
		\small
			\zeta_{\vartheta}\left( {r,h_{\mathrm{u}}} \right) = \eta_{\vartheta}\left\lbrack {r^{2} + \left( h_{\mathrm{u}} - h_{\mathrm{b}}\right) ^{2}} \right\rbrack^{- \frac{\alpha_{\vartheta}}{2}}
		\end{equation}
		
		where $\vartheta \in \left\{ \mathrm{L},\mathrm{N} \right\}$ with $\mathrm{L}$ and $\mathrm{N}$ , respectively, referring to the LoS and NLoS transmissions, and $r$ is the horizontal distance from UAV-UE to its serving GBS. Due to the 3D mobility model, when the UAV-UE is located at $X_{t-1}$ and $X_{t}$ at the $t$-th movement, $h_{\mathrm{u}}$ refers to $Z_{t-1}$ and $Z_{t}$, respectively. Similarly, $\eta_{\mathrm{L}}$ and $\eta_{\mathrm{N}}$ are the path loss constants at the reference distance $d=1\mathrm{m}$ for the LoS and NLoS components, respectively. 

		We adopt the Nakagami-$m_{\vartheta}$ model for the small-scale fading with the PDF given by $f_{\Omega }\left( \omega \right) = \frac{m_{\vartheta}^{m_{\vartheta}}\omega^{m_{\vartheta} - 1}}{\Gamma\left( m_{\vartheta} \right)}e^{- m_{\vartheta}\omega}$, where $m_{\mathrm{L}}$, $m_{\mathrm{N}}$ with $m_{\mathrm{L}}>m_{\mathrm{N}}$ represent the fading parameters of LoS and NLoS transmissions respectively, the channel gain power $\omega = x^{2}\sim\mathrm{Gamma} \left ( m_{\vartheta},1/m_{\vartheta} \right ) $ when $x\sim$ Nakagami-$m_{\vartheta}$ and $\Gamma\left( . \right)$ is the Gamma function \cite{ref10}.

		For the A2G channel model, the LoS probability is not only related to the environment, including the density and height of buildings, but also the elevation angle of the UAV-UE to the GBS, which is given by\cite{ref5}
		\begin{equation}
		\small
			P_{\mathrm{L}}\left( r,h_{\mathrm{u}} \right) = \frac{1}{1 + \mathrm{a} \exp \left( {- \mathrm{b}\left( {\frac{180}{\pi}{\mathrm{tan}}^{- 1}\left( \frac{h_{\mathrm{u}} - h_{\mathrm{b}}}{r} \right) - \mathrm{a}} \right)} \right)}
		\end{equation}
		where $\mathrm{a}$ and $\mathrm{b}$ are parameters related to the environment. It can be seen that the LoS probability increases with $h_{\mathrm{u}}$ and decreases with $r$. Therefore, the NLoS probability can be expressed as $P_{\mathrm{N}}\left( r,h_{\mathrm{u}} \right) = 1 - P_{\mathrm{L}}\left( r,h_{\mathrm{u}} \right)$.

	\subsection{Cell Association and SIR}

		In this letter, we consider the strongest average RSS association strategy which refers to the case that the UAV-UE associates with the GBS that provides the strongest average received signal power.

		To be specific, the received signal power of the typical UAV-UE located at a horizontal distance $r_{0}$ away from its serving GBS $g_{0}$ is given by
		\begin{equation}
		\small
			\mathrm{P_{r}}\left ( r_{0} \right )  = \mathrm{P_{t}G_{tot}}\zeta_{\vartheta} \left ( r_{0},h_{\mathrm{u}} \right ) \Omega_{\vartheta},\ \vartheta \in \left\{ \mathrm{L},\mathrm{N} \right\}
		\end{equation}
		where $\Omega_{\vartheta}$ denotes the channel power gain when associating with a $\vartheta$-type serving GBS. Therefore, the signal-to-interference ratio (SIR) of the typical UAV-UE is expressed as
		\begin{equation}
		\small
			\mathrm{SIR}=\frac{\mathrm{P_{r}}\left ( r_{0} \right ) }{I} =\left\{\begin{matrix} \frac{\mathrm{P_{t}G_{tot}}\zeta_{\mathrm{L}}\left ( r_{0},h_{\mathrm{u}} \right )\Omega _{\mathrm{L}}}{ { \sum_{i\in \Phi \setminus \left\lbrace g_{0}\right\rbrace }\mathrm{P_{r}}\left ( r_{i} \right ) } } \ ;\ \mathrm{for\ LoS}\ \ \ \\ \frac{\mathrm{P_{t}G_{tot}}\zeta_{\mathrm{N}}\left ( r_{0},h_{\mathrm{u}} \right )\Omega _{\mathrm{N}}}{ { \sum_{i\in \Phi \setminus \left\lbrace g_{0}\right\rbrace }\mathrm{P_{r}}\left ( r_{i} \right ) } } \ ;\ \mathrm{for\ NLoS} \end{matrix}\right.
		\end{equation}
		Note that we consider the interference-limited region and ignore the noise power in (4), where $g_{0}$ denotes the serving GBS of the typical UAV-UE.

	\section{HANDOVER ANALYSIS}

		Referring to Fig. \ref{fig_1}, equipped with the directional antenna with beamwidth $\varphi_{b}$, the horizontal projection of the signal/interference receiving range of the typical UAV-UE is limited to a circle $\mathcal{O}\left (q_{t}, r_\mathrm{M} \right )$ with $q_{t}$ being the horizontal position of the UAV-UE at the end way-point of the $t$-th movement, and $r_\mathrm{M}$ being the radius. We assume that $H_{\vartheta}^{\varsigma}$ represents the event that the UAV-UE handovers from a $\varsigma$-type serving GBS to a $\vartheta$-type target GBS for $\vartheta,\varsigma \in \left \{ \mathrm{L,N} \right \}$.

		\begin{lemma}\label{Lemma1}
		Given the $\varsigma$-type original serving GBS, $r_0$ and $z_t$, the conditional handover probability of the typical UAV-UE equipped with directional antenna is given by (5) (as shown at the top of the page).
		\begin{figure*}
		\vspace{-2em}
			\begin{equation}
			\small
				\begin{aligned}
				\mathbb{P}\left(H_{\vartheta}^{\varsigma} \mid r_{0},z_{t}\right)=\begin{cases}
				1 - \frac{1}{\pi} \int_{0}^{\pi} \int_{0}^{\infty }\int_{h_{\mathrm{lb}}}^{h_{\mathrm{ub}}} {\mathit{\exp}\left( {- \lambda_{\mathrm{b}} F \left ( D_{\vartheta}^{\varsigma }\left( r_{0} \right) , D_{\vartheta}^{\varsigma }\left( R \right) \right )}\right)}f_{Z_{t-1}}\left ( z_{t-1}  \right )f_{\rho_{t}}\left ( \varrho_{t}  \right ) dz_{t-1}d\varrho_{t} d\theta ,&{\text{if}\ D_{\vartheta}^{\varsigma }\left( R \right)\le r_{\mathrm{M}}}\\
				1 - \frac{1}{\pi} \int_{0}^{\pi} \int_{0}^{\infty }\int_{h_{\mathrm{lb}}}^{h_{\mathrm{ub}}}{\mathit{\exp}\left( {- \lambda_{\mathrm{b}} F \left ( D_{\vartheta}^{\varsigma }\left( r_{0} \right) , r_{\mathrm{M}} \right )}\right)}f_{Z_{t-1}}\left ( z_{t-1}  \right )f_{\rho_{t}}\left ( \varrho_{t}  \right )dz_{t-1}d\varrho_{t} d\theta , &{\text{if}\ 0<r_{\mathrm{M}}<D_{\vartheta}^{\varsigma }\left( R \right)}
				\end{cases}
				\end{aligned}
			\end{equation}
		\vspace{-2em}
		\end{figure*}
	
	 	\noindent 
		where $ R=\sqrt{r_{0}^{2} + V_{\mathrm{h}}^{2} + 2r_{0}V_{\mathrm{h}}\cos \theta}$ is the horizontal distance between the typical UAV-UE and the original serving GBS $g_{0}$ when moving to $X_{t}$, $V_{h}$ is the horizontal movement length within the unit time, $r_{\mathrm{M}} = \bar{h}{\mathit{\tan}\left( \frac{\varphi_{\mathrm{b}}}{2} \right)}$ with $\bar{h} = z_{t} - h_{\mathrm{b}}$, $D_{\vartheta}^{\varsigma }\left( x \right) =\left\{\begin{matrix} x & \vartheta = \varsigma   \\ \sqrt{\left( \frac{\eta_{\varsigma }}{\eta_{\vartheta}} \right)^{\frac{2}{\alpha_{\varsigma }}}\left( {x^{2} + {\bar{h}}^{2}} \right)^{\frac{\alpha_{\vartheta}}{\alpha_{\varsigma }}} - {\bar{h}}^{2}} & \vartheta \neq \varsigma \end{matrix}\right.$, and $F \left ( x,y \right ) = y^{2}{\left( \pi -  \mathit{\cos}^{- 1}\left( \frac{y^{2} + V_{\mathrm{h}}^{2} - x^{2}}{2yV_{\mathrm{h}}} \right)\right) }-x^{2}{\mathit{\cos}^{- 1}\left( \frac{x^{2} + V_{\mathrm{h}}^{2} - y^{2}}{2xV_{\mathrm{h}}} \right)}+\frac{1}{2}\sqrt{\left [ \left ( x+V_{\mathrm{h}} \right )^{2}-y^{2} \right ]\left [ y^{2}-\left ( x-V_{\mathrm{h}} \right )^{2} \right ] }$.
		\end{lemma}
	
		\begin{proof}
		See Appendix A.
		\end{proof}
		\vspace{-0.4em} 
		From (5), we observe that $\mathbb{P}\left(H_{\vartheta}^{\varsigma} \mid r_{0},z_{t}\right)$ is closely related to $\varphi_{\mathrm{b}}$, $\lambda_{\mathrm{b}}$ and $V_{\mathrm{h}}$. A higher handover probability occurs when the vertical movements and horizontal transition change more dramatically, especially for a denser network. 
 
		We then derive the conditional handover probability of the 3D mobile typical UAV-UE when switching over to a $\vartheta$-type target GBS
		\begin{equation}
		\small
			\mathbb{P}\left(H_{\vartheta} \mid r_{0},z_{t}\right)= {\sum\limits_{\varsigma \in \mathrm{\{L,N\}}}} \mathbb{P}\left( H_{\vartheta}^{\varsigma} \mid  r_{0},z_{t}\right)A_{\varsigma}\left( z_{t} \right)
		\end{equation}
		where $A_{\varsigma}\left( z_{t} \right)$ is the probability that the typical UAV-UE associates with a $\varsigma$-type GBS, given by
		\begin{equation}
		\small
			A_{\varsigma}\left( z_{t} \right) = {\int_{0}^{r_{\mathrm{M}}}{\exp\left( - 2\pi\lambda_{\mathrm{b}}{\int_{0}^{r_{\xi}^{\varsigma}}xP_{\xi}\left( x, z_{t} \right)dx }\right)}}f_{R_{0}^{\varsigma}}\left( r_{0},z_{t} \right)dr_{0}
		\end{equation}
		where $\varsigma \neq \xi \in \left\{ \mathrm{L},\mathrm{N} \right\}$ with $r\mathrm{_{L}^{N}} = \mathrm{min}\left( \rho\mathrm{_{N-L}}, r_{\mathrm{M}} \right)$, $r\mathrm{_{N}^{L}} = \mathrm{max}\left( \rho\mathrm{_{L-N}}, 0 \right)$. Given that the typical UAV-UE associates with a LoS (NLoS) GBS, the nearest NLoS (LoS) GBS is at least at distance $\rho_{\mathrm{L-N}} = \sqrt{\left( \frac{\eta_{\mathrm{N}}}{\eta_{\mathrm{L}}} \right)^{\frac{2}{\alpha_{\mathrm{N}}}}\left( {r_{0}^{2} + {\bar{h}}^{2}} \right)^{\frac{\alpha_{\mathrm{L}}}{\alpha_{\mathrm{N}}}} - {\bar{h}}^{2}}$ $\left( \rho_{\mathrm{N-L}} = \sqrt{\left( \frac{\eta_{\mathrm{L}}}{\eta_{\mathrm{N}}}\right)^{\frac{2}{\alpha_{\mathrm{L}}}}\left( {r_{0}^{2} + {\bar{h}}^{2}}\right)^{\frac{\alpha_{\mathrm{N}}}{\alpha_{\mathrm{L}}}} -{\bar{h}}^{2}}\right) $. $f_{R_{0}^{\varsigma}}\left( r_{0},z_{t} \right) = 2\pi\lambda_{\mathrm{b}}r_{0}P_{\varsigma}\left( r_{0}, z_{t} \right) \exp \left( - 2\pi\lambda_{\mathrm{b}}{\int_{0}^{r_{0}}xP_{\varsigma}\left( x, z_{t} \right)dx} \right)$ represents the PDF of the distance between the typical UAV-UE and the nearest $\varsigma$-type GBS.

	\section{COVERAGE PROBABILITY ANALYSIS}

		The coverage probability is defined as the probability that the SIR of the typical UAV-UE is larger than a given threshold $\mathrm{T}$, which can be expressed as
		\begin{equation}
		\small
			P_{\mathrm{cov}} = \mathbb{P}\left( {\text{SIR} > \mathrm{T},\bar{H}} \right) + \left( {1 - \kappa} \right)\mathbb{P}\left( {\text{SIR} > \mathrm{T}}, H \right)
			\vspace{-0.5em}
		\end{equation}
		where $\kappa$ is the connection failure probability \cite{ref8} caused by the handover, and $\bar{H}$ indicates the event that no handover occurs.

		\begin{theorem}
			The coverage probability of the typical 3D mobile UAV-UE equipped with directional antenna is given by
			\begin{equation}
			\small
				\begin{aligned}
					&P_{\mathrm{cov}} = \\&\left( {1 - \kappa} \right)\sum_{\vartheta \in \{ \mathrm{L,N}\}} {\int_{h_{\mathrm{lb}}}^{h_{\mathrm{ub}}}{{\int_{0}^{r_{\mathrm{M}}}     A_{\vartheta}\left ( z_{t} \right )\mathbb{P}_{\mathrm{cov}}^{\vartheta}\left( r_{0},z_{t} \right) }}}\tilde{f}_{R_{0}^{\vartheta}}\left( r_{0},z_{t} \right) \\& \times f_{Z_{t}}\left( z_{t}\right)dr_{0}dz_{t} +\kappa\sum_{\vartheta \in \{ \mathrm{L,N}\}}{\int_{h_{\mathrm{lb}}}^{h_{\mathrm{ub}}}{{\int_{0}^{r_{\mathrm{M}}}A_{\vartheta}\left ( z_{t} \right ){\mathbb{P}\left(\bar{H}_{\vartheta} \mid r_{0},z_{t}\right)}}}} \\& \times \mathbb{P}_{\mathrm{cov}}^{\vartheta}\left( r_{0},z_{t} \right)\tilde{f}_{R_{0}^{\vartheta}}\left( r_{0},z_{t}\right) f_{Z_{t}}\left( z_{t}\right) dr_{0}dz_{t}
				\end{aligned}
			\vspace{-0.5em}
			\end{equation}
			where $\tilde{f}_{R_{0}^{\vartheta}}\left( r_{0},z_{t} \right) = \frac{f_{R_{0}^{\vartheta}}\left( r_{0},z_{t} \right)}{A_{\vartheta}\left ( z_{t} \right )}{\exp\left( - 2\pi\lambda_{\mathrm{b}}{\int_{0}^{r_{\xi}^{\vartheta}}xP_{\xi}\left( x, z_{t} \right)dx }\right)} $ represents the PDF of the distance between the typical UAV-UE and the serving $\vartheta$-type GBS with $\vartheta \neq \xi \in \left\{ \mathrm{L},\mathrm{N} \right\}$. $\mathbb{P}_{\mathrm{cov}}^{\vartheta}\left( r_{0},z_{t} \right)$ denotes the conditional coverage probability, given by
			\begin{equation}
			\small
				\mathbb{P}_{\mathrm{cov}}^{\vartheta}\left( r_{0},z_{t} \right) = {\sum_{l = 0}^{m_{\vartheta} - 1}\frac{\left( {- \tau_{\vartheta}} \right)^{l}}{l!}}\frac{d^{l}}{d\tau_{\vartheta}^{l}}\mathcal{L}_{I|R_{0},Z_{t}}\left( \tau_{\vartheta} \right)
				\vspace{-0.5em}
			\end{equation}
			in which $\mathcal{L}_{I|R_{0},Z_{t}}\left( \tau_{\vartheta} \right) = \exp\left( - 2\pi\lambda_{\mathrm{b}}{\sum\limits_{\xi \in \{ \mathrm{L},\mathrm{N}\}}{\int_{r_{\xi}^{\vartheta}}^{r_{\mathrm{M}}}{P_{\xi}\left( x,z_{t} \right)\gamma_{\xi}\left( x,\tau_{\vartheta},z_{t} \right)xdx}}} \right)$ with $r\mathrm{_{L}^{L}} = r\mathrm{_{N}^{N}} = r_{0}$, $r\mathrm{_{L}^{N}} = \mathrm{min}\left( \rho\mathrm{_{N-L}}, r_{\mathrm{M}} \right)$, $r\mathrm{_{N}^{L}} = \mathrm{max}\left( \rho\mathrm{_{L-N}}, 0 \right)$, $\gamma_{\xi}\left( x,\tau_{\vartheta},z_{t} \right) = 1 - \left( \frac{m_{\xi}}{m_{\xi} + \tau_{\vartheta}\mathrm{P_{t}G_{tot}}\zeta_{\xi}\left( x,z_{t} \right)} \right)^{m_{\xi}}$, and $\tau_{\vartheta} = \frac{m_{\vartheta}\mathrm{T}}{\mathrm{P_{t}G_{tot}}\zeta_{\vartheta}\left( {r_{0},z_{t}} \right)}$, respectively.
		\end{theorem}
	
		\begin{proof}
			The result can be proved by the modification of Theorem 1 in \cite{ref11} by considering the directional antenna and the handover event, where the former limits the interfering GBSs range and the latter affects the connection failure probability. The complete proof is omitted due to the space limitation.
		\end{proof}
		\vspace{-0.4em} 
		It is worth noting that the nearest association strategy can be seen as a special case of the strongest average RSS strategy, where the GBS type is neglected for the association and handover.
		\begin{corollary}
			When the nearest association strategy is adopted, the coverage probability of the typical 3D mobile UAV-UE equipped with directional antenna is given by
			\begin{equation}
			\small
			\begin{aligned}
					&P_{\mathrm{cov}}^{\mathrm{n}} = \\&\left( {1 - \kappa} \right)\sum_{\vartheta \in \{ \mathrm{L,N}\}} {\int_{h_{\mathrm{lb}}}^{h_{\mathrm{ub}}}{{\int_{0}^{r_{\mathrm{M}}}P_{\vartheta}\left( r_{0}, z_{t} \right) \mathbb{P}_{\mathrm{cov|n}}^{\vartheta}\left( r_{0},z_{t} \right) }}}f_{R_{0}}^{\mathrm{n}}\left( r_{0} \right) \\& \times f_{Z_{t}}\left( z_{t}\right)dr_{0}dz_{t} +\kappa\sum_{\vartheta \in \{ \mathrm{L,N}\}}{\int_{h_{\mathrm{lb}}}^{h_{\mathrm{ub}}}{{\int_{0}^{r_{\mathrm{M}}}P_{\vartheta}\left( r_{0}, z_{t} \right)}}} \\& \times {\mathbb{P}_{\mathrm{n}}\left(\bar{H} \mid r_{0},z_{t}\right)} \mathbb{P}_{\mathrm{cov|n}}^{\vartheta}\left( r_{0},z_{t} \right)f_{R_{0}}^{\mathrm{n}}\left( r_{0}\right) f_{Z_{t}}\left( z_{t}\right) dr_{0}dz_{t}
				\end{aligned}
			\end{equation}
			where $\mathbb{P}_{\mathrm{cov|n}}^{\vartheta}\left( r_{0},z_{t} \right)$ is similar to (10) with $\mathcal{L}_{I|R_{0},Z_{t}}^{\mathrm{n}}\left( \tau_{\vartheta} \right) = \exp\left( - 2\pi\lambda_{\mathrm{b}}{\sum\limits_{\xi \in \{ \mathrm{L},\mathrm{N}\}}{\int_{r_{0}}^{r_{\mathrm{M}}}{P_{\xi}\left( x,z_{t} \right)\gamma_{\xi}\left( x,\tau_{\vartheta},z_{t} \right)xdx}}} \right)$, $f_{R_{0}}^{\mathrm{n}}\left( r_{0} \right) = 2\pi\lambda_{\mathrm{b}}r_{0} \exp \left( - \pi \lambda_{\mathrm{b}} r_{0}^{2} \right)$, and $\mathbb{P}_{\mathrm{n}}\left(H \mid r_{0},z_{t}\right)$ is derived by setting $\vartheta = \varsigma$ in (5).
		\end{corollary}
	
		\begin{proof}
			For the nearest association strategy, $\mathbb{P}_{\mathrm{cov|n}}^{\vartheta}\left( r_{0},z_{t} \right)$, $\mathbb{P}_{\mathrm{n}}\left(H_{\vartheta} \mid r_{0},z_{t}\right)$ and $f_{R_{0}}^{\mathrm{n}}\left( r_{0} \right)$ can be derived by analogy with Theorem 1, which simplifies the association and handover analysis. The complete proof is omitted due to the space limitation. 
		\end{proof}
		
		\vspace{-1em}
		\section{SIMULATION RESULTS AND ANALYSIS}

		In this section, we verify the theoretical analysis via extensive simulations, and show the superiority of the directional antenna technology compared to that with omnidirectional antenna. Unless otherwise specified, we adopt the following default values: $\varphi_{\mathrm{b}} =120^{\circ}$, $\mathrm{\lambda_{b}=100GBSs/km^{2}}$, $\mathrm{T=-3.8dB}$, $m_{\mathrm{L}} = 3$, $m_{\mathrm{N}} = 1$, $\alpha_{\mathrm{L}} = 2.09$, $\alpha_{\mathrm{N}} = 3.75$, $\eta_{\mathrm{L}} = -41.1\mathrm{dB}$, $\eta_{\mathrm{N}} = -32.9\mathrm{dB}$, $\kappa = 0.3$, $\mathrm{P_{t}} = 46\mathrm{dBm}$, $\mathrm{a} = 9.61$, $\mathrm{b} = 0.16$, $\mu = 300 km^{-2}$, $h_{\mathrm{b}} = 30\mathrm{m}$, $h_{\mathrm{lb}} = 90\mathrm{m}$, and $h_{\mathrm{ub}} = 150\mathrm{m}$.

		\begin{figure}[t]
			\begin{center}
			\flushleft
				\subfigure[]{
					\begin{minipage}{0.45\textwidth}
						\centering
						\includegraphics[width = 1\textwidth]  {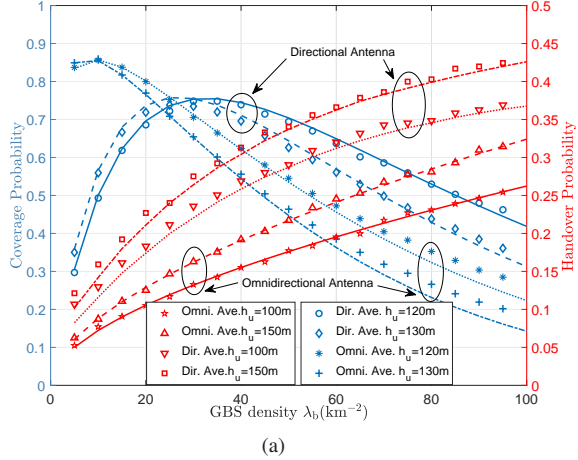}
					\end{minipage}
				}
				\subfigure[]{
					\begin{minipage}{0.45\textwidth}
						\centering
						\includegraphics[width = 1\textwidth]  {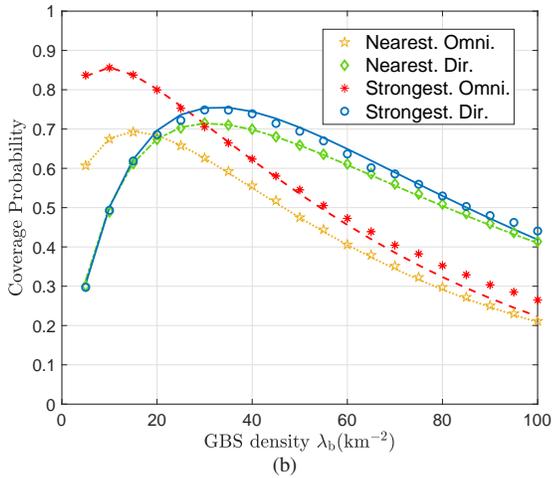}
					\end{minipage}
				}
				\vspace{-0.2cm}
				\captionsetup{font={scriptsize}}
				\caption{
					(a) Coverage probability and handover probability vs GBS density $\mathrm{\lambda_{b}}$ with strongest average RSS association, (b) coverage probability vs GBS density $\mathrm{\lambda_{b}}$ for the comparison with two association strategies.
				}
				\label{fig_2}
				\vspace{-0.5cm}
			\end{center}
		\end{figure}

		In Fig. \ref{fig_2}(a), we plot the coverage probability and handover probability as a function of GBS density under the strongest average RSS association for $\mathrm{T=-3.8dB}$. We observe that the handover probability grows with the increasing GBS density and the height of UAV-UE. This is because the UAV-UE sees more GBSs in the above conditions, which brings extra handover opportunities. Compared with the omnidirectional antenna model, equipping with the directional antenna model limits the number of potential GBSs, leading to a considerable reduction (more than 0.1 in this case) in handover probability. When $\mathrm{\lambda_{b}<20GBSs/km^{2}}$, the coverage probability of UAV-UE with the directional antenna is smaller than that with the omnidirectional antenna. This is because when the GBS density is extremely small, there is the possibility that no GBSs exist within the receiving range of the typical UAV-UE with directional antenna, leading to the decline of coverage probability. We observe that there exists an optimal GBS density that maximizes the coverage probability due to the incremental interference. What's more, it is shown that a higher average height $Z_{t}$ results in a lower coverage probability in most of $\lambda_{\mathrm{b}}$, while the opposite exists when $\lambda_{\mathrm{b}}$ is smaller. This is due to the compromise between the incremental LoS probability and the increased path loss at a larger altitude for different GBS densities. 

		In Fig. \ref{fig_2}(b), the coverage probabilities achieved by strongest average RSS and nearest association strategies are compared. We observe different phenomenon when UAV-UEs are equipped with omnidirectional attenna and directional antenna, respectively. With omnidirectional attenna, the superiority of the strongest average RSS association over the counterpart diminishes as the GBS density increases. This is because as GBS density increases, the probability that both association strategies lead to the connection with the same GBS improves. When UAV-UEs are equipped with directional attenna, the gap between the two strategies first increases, and then keeps unchanged with the increase in GBS density. This can be caused by the nonexistence of GBSs within the interference receiving range of UAV-UEs when the GBS density is small. 

		\begin{figure}[t]
			\flushleft
			\subfigure[]{
				\begin{minipage}{0.45\textwidth}
					\centering
					\includegraphics[width = 1\textwidth]  {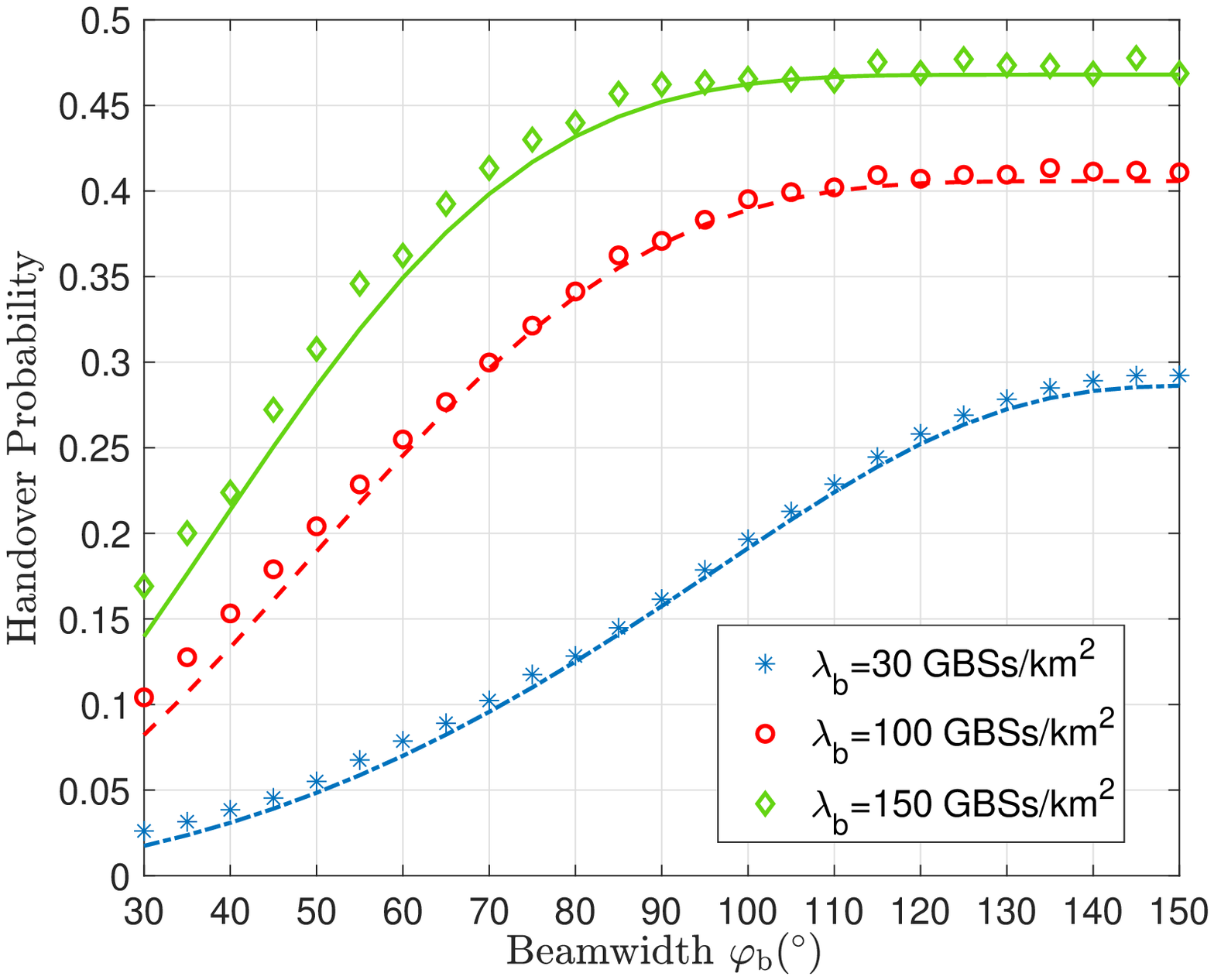}
				\end{minipage}
			}
			\subfigure[]{
				\begin{minipage}{0.45\textwidth}
					\centering
					\includegraphics[width = 1\textwidth]  {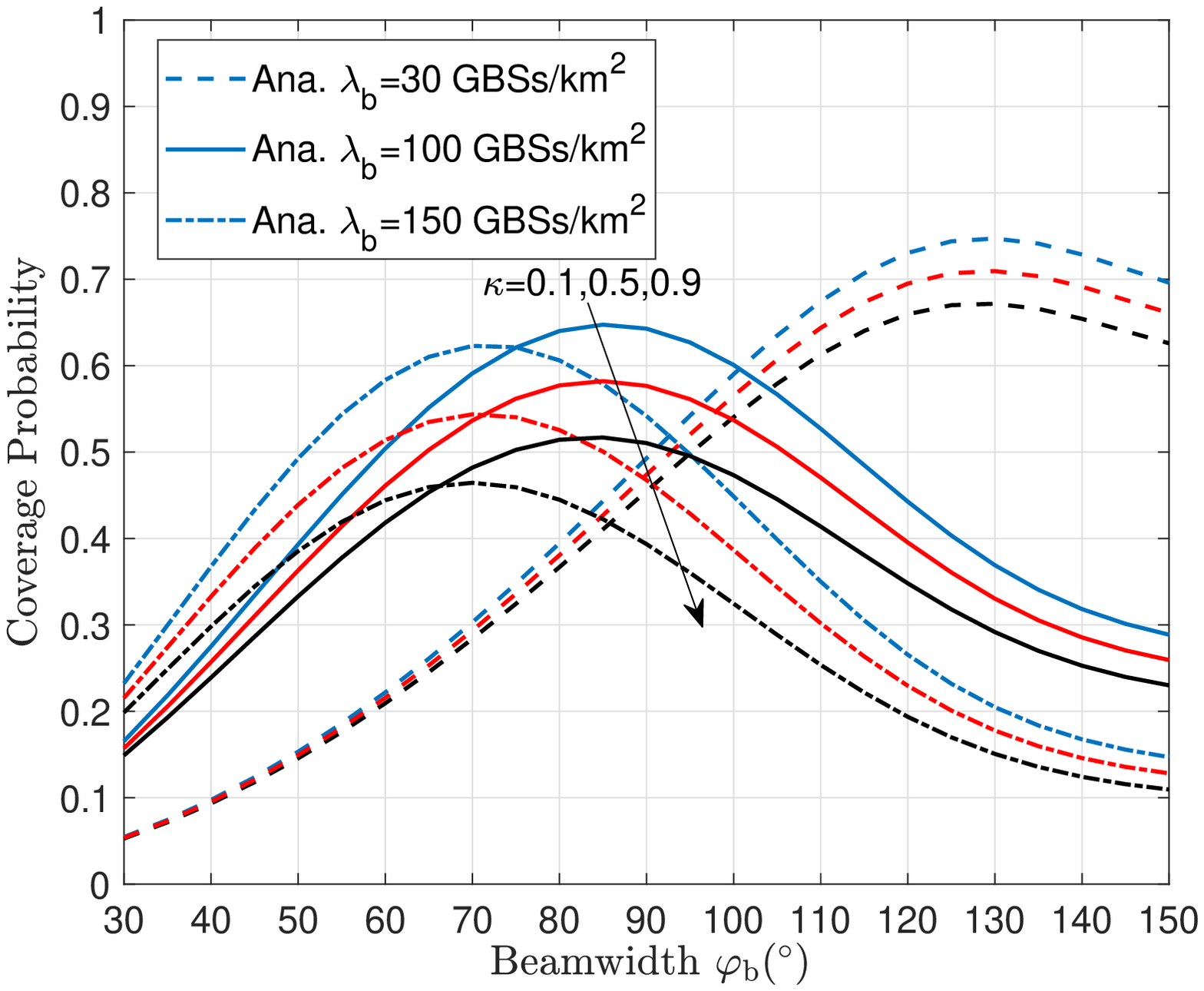}
				\end{minipage}
			}
			\vspace{-0.2cm}
			\captionsetup{font={scriptsize}}
			\caption{
				The impacts of antenna beamwidth on (a) handover probability, and (b) coverage probability, under the strongest average RSS association for different GBS densities $\lambda_\mathrm{b}$.
			}
			\label{fig_3}
			\vspace{-0.25cm}
		\end{figure}

		In Fig. \ref{fig_3}(a), we depict the handover probability as a function of antenna beamwidth for different GBS densities $\lambda_\mathrm{b}$ under the strongest average RSS association strategy. We observe that a higher $\lambda_\mathrm{b}$ leads to a larger handover probability. What's more, as the antenna beamwidth $\varphi_\mathrm{b}$ increases, the handover probability first grows and then converges. This is because both increasing $\lambda_\mathrm{b}$ and $\varphi_\mathrm{b}$ enlarge the number of potential target GBSs, which enhances the opportunities to switch over the strongest target GBS in the whole network during the movement. However, the strongest target GBS has fewer probability to lie within the area far away from the initial serving GBS, which explains the convergence of handover probability with regards to $\varphi_\mathrm{b}$.
		
		In Fig. \ref{fig_3}(b), we depict the coverage probability as a function of antenna beamwidth $\varphi_\mathrm{b}$ for different GBS densities $\lambda_\mathrm{b}$ and connection failure probability $\kappa$ under the strongest average RSS association. We observe that as $\varphi_b$ enlarges, the coverage probability first rises and then declines. A larger $\lambda_\mathrm{b}$ results in a higher coverage for a smaller beamwidth. On the contrary, a smaller $\varphi_\mathrm{b}$ leads to a higher coverage probability for a larger $\varphi_\mathrm{b}$. Both of the above can be explained by the tradeoff between the more preferred targeting GBSs and the extra aggregated interference due to the incremental receiving range of UAV-UEs. What's more, a larger $\kappa$ results in a lower coverage probability due to the handover induced connection failure. 
	\vspace{-1em}
	\section{Conclusion}

		In this letter, we proposed an analytical framework to evaluate the impact of directional antenna technology on the cellular-connected UAV network under the strongest average RSS association strategy, in which UAV-UEs move according to a 3D mobility model. The exact analytical expressions for handover probability and coverage probability were derived. It was concluded that the optimal UAV-UE antenna beamwidth decreases with the GBS density, where the omnidirectional antenna model was preferred in the sparse network scenario. In addition, the strongest average RSS association gradually degenerates to the nearest association with the growth of GBS density. 

	\vspace{-1.5em}
	\begin{appendices}
		\section{ PROOF OF LEMMA 1 }
			From Fig. \ref{fig_1}, a handover event occurs when there exists at least one GBS in the area $\mathcal{B} \setminus \mathcal{A} \cap \mathcal{B}$. Hence, the handover probability is given by
			\begin{equation}
			\small
				\begin{aligned}
					\mathbb{P}\left (H_{\vartheta}^{\varsigma} \mid \right.&\left.r_{0},\theta, \varrho_{t},z_{t},z_{t-1} \right ) \\&= \mathbb{P}\left ( N \left( \mathcal{B} \setminus \mathcal{A} \cap \mathcal{B} \right) > 0 \mid r_{0},\theta,\varrho_{t},z_{t},z_{t-1} \right ) \\&\overset{\mathrm{\left( a \right) }}{=} 
					1 - \exp\left ( -\lambda_{b} \left( \left | \mathcal{B} \right | - \left | \mathcal{A} \cap \mathcal{B} \right | \right) \right )
				\end{aligned}
			\vspace{-0.5em}
			\end{equation}
			where $N\left ( \cdot  \right ) $ is the number of GBSs in the specified area, and (a) results from the null probability of a 2D PPP. 
	
			With the strongest average RSS association strategy, when UAV-UE is located at $X_{t-1}$, $\mathcal{A}=\mathcal{O}\left( q_{t-1}, D_{\vartheta}^{\varsigma }\left( r_{0} \right)\right) $ denotes the circle with $\mathrm{q_{t-1}}$ being the center and $D_{\vartheta}^{\varsigma }\left( r_{0} \right)$ being the corresponding radius, which means that the serving GBS located at $g_{0}$ lies within the signal/interference receiving range. Once the UAV-UE is flying to $X_{t}$, there may exist the following two cases: 
			\begin{itemize}
				\item[(a)] When the height of the UAV-UE becomes higher and the signal/interference receiving range is larger, i.e., $D_{\vartheta}^{\varsigma }\left( R \right) \le r_{\mathrm{M}}$, we have $\mathcal{B}=\mathcal{O}\left( q_{t}, D_{\vartheta}^{\varsigma }\left( R \right) \right)$, where $D_{\vartheta}^{\varsigma }\left( R \right)$ denotes the horizontal distance between the nearest interfering $\vartheta$-type GBS and the typical UAV-UE whose serving GBS $g_{0}$ is $\varsigma$-type at the $t$-th moment.
				\item[(b)] When the height of the UAV-UE becomes lower and the signal/interference receiving range is smaller, i.e., $0<r_{M}<D_{\vartheta}^{\varsigma }\left( R \right)$, we have $\mathcal{B}=\mathcal{O}\left( q_{t}, r_{\mathrm{M}} \right) $. 
			\end{itemize}
		
			According to Heron's formula and trigonometric function, we derive $\left | \mathcal{O}\left (q_{t}, y \right ) \right | - \left | \mathcal{O}\left (q_{t}, y \right ) \cap \mathcal{O}\left (q_{t-1}, x\right ) \right | = F \left ( x ,y \right )$ which is given in Lemma \ref{Lemma1}.
			What's more, by averaging over $\Theta$, $\rho_{t}$, and $Z_{t-1}$ with PDFs being $f_{\Theta}\left( \theta \right) $, $f_{\rho_{t}}\left( \varrho_{t} \right) $ and $f_{Z_{t-1}}\left( z_{t-1} \right) = \frac{1}{h_{\mathrm{ub}}-h_{\mathrm{lb}}}$, the resulting conditional handover probability $\mathbb{P}\left(H_{\vartheta}^{\varsigma} \mid r_{0},z_{t}\right)$ is given by (5). This completes the proof.
			
	\end{appendices}


\begin{thebibliography}{1}
		\bibliographystyle{IEEEtran}
		\bibitem{ref1}
		M. Mozaffari, W. Saad, M. Bennis, Y.-H. Nam, and M. Debbah, “A tutorial on UAVs for wireless networks: Applications, challenges, and open problems,” \emph{IEEE Commun. Surveys Tuts.}, vol. 21, no. 3, pp. 2334–2360, 3rd Quart., 2019.
		
		\bibitem{ref2}
		Y. Zeng, J. Lyu, and R. Zhang, “Cellular-connected UAV: Potential, challenges, and promising technologies,” \emph{IEEE Wirel. Commun.}, vol. 26, no. 1, pp. 120–127, Feb. 2019.
		
		\bibitem{ref3}
		X. Lin et al., “The sky is not the limit: LTE for unmanned aerial vehicles,” \emph{IEEE Commun. Mag.}, vol. 56, no. 4, pp. 204–210, Apr. 2018.
		
		\bibitem{ref4}
		M. M. Azari, F. Rosas and S. Pollin, ”Cellular Connectivity for UAVs: Network Modeling, Performance Analysis, and Design Guidelines,” \emph{IEEE Trans. Wirel. Commun.}, vol. 18, no. 7, pp. 3366-3381, July 2019.
		
		\bibitem{ref5}
		P. K. Sharma and D. I. Kim, “Random 3D mobile UAV networks: Mobility modeling and coverage probability,” \emph{IEEE Trans. Wirel. Commun.}, vol. 18, no. 5, pp. 2527–2538, May 2019.
		
		\bibitem{ref6}
		S. Enayati, H. Saeedi, H. Pishro-Nik, and H. Yanikomeroglu, “Moving aerial base station networks: A stochastic geometry analysis and design perspective,” \emph{IEEE Trans. Wirel. Commun.}, vol. 18, no. 6, pp. 2977–2988, Jun. 2019.
		
		\bibitem{ref7}
		M. Banagar and H. S. Dhillon, “Performance characterization of canonical mobility models in drone cellular networks,” \emph{IEEE Trans. Wirel. Commun.}, vol. 19, no. 7, pp. 4994-5009, Jul. 2020.
		
		\bibitem{ref8}
		R. Amer, W. Saad and N. Marchetti, ”Mobility in the Sky: Performance and Mobility Analysis for Cellular-Connected UAVs,” \emph{IEEE Trans. Commun.}, vol. 68, no. 5, pp. 3229-3246, May 2020.
		
		\bibitem{ref9}
		Further Advancements for E-UTRA Physical Layer Aspects(Release 9), document 3GPP TR 36.814, Mar. 2010. [Online]. Available:
		www.qtc.jp/3GPP/Specs/36814-900.pdf
		
		\bibitem{ref10}
		M. Alzenad and H. Yanikomeroglu, "Coverage and Rate Analysis for Vertical Heterogeneous Networks (VHetNets)," \emph{IEEE Trans. Wirel. Commun.}, vol. 18, no. 12, pp. 5643-5657, Dec. 2019.
		
		\bibitem{ref11}
		H. Sun, X. Wang, Y. Zhang and T. Q. S. Quek, "Performance Analysis and Cell Association Design for Drone-Assisted Heterogeneous Networks," \emph{IEEE Trans. Veh. Technol.}, vol. 69, no. 11, pp. 13741-13755, Nov. 2020.
		
	\end{thebibliography}
\end{document}